\documentclass[journal]{IEEEtran}
\usepackage{amsmath}
\usepackage{indentfirst}
\usepackage{latexsym}
\usepackage{graphicx}
\usepackage{bm}
\usepackage{amssymb}
\usepackage{array}
\usepackage{setspace}
\usepackage{fancyhdr}
\usepackage{float}
\usepackage{algorithm}
\usepackage{algorithmicx}
\usepackage{algpseudocode} 
\usepackage{epsfig}
\usepackage{epstopdf}

\newcommand{\subparagraph}{}
\usepackage{titlesec}
\usepackage{balance}
\usepackage{hyperref}

\usepackage{csquotes}
\MakeOuterQuote{"}

\usepackage{caption}
\newcommand{\ignore}[1]{}

\makeatletter
\newcommand{\vast}{\bBigg@{3.6}}
\newcommand{\Vast}{\bBigg@{4.6}}
\makeatother

\newtheorem{theorem}{Theorem}

\usepackage{parskip}
\titlespacing*{\section} {0pt}{2ex plus 1ex minus .2ex}{1.5ex plus .2ex}

\begin{document}

%
\title{\vspace{-.5cm} Throughput Maximization in Cloud Radio Access Networks using Network Coding}

\author{
{Mohammed S. Al-Abiad$^\ast$, Ahmed Douik$^{\dagger}$, Sameh Sorour$^\ddagger$, and MD Jahangir Hossain$^\ast$\\}%
{$^\ast$ The University of British Columbia, Kelowna, BC, Canada\\
$^\dagger$California Institute of Technology, Pasadena, CA, United States of America \\
$^\ddagger$ University of Idaho, Moscow, ID, United States of America \\
$^\ast$mohammed.saif@mail.ubc.ca,jahangir.hossain@ubc.ca $^\dagger$ahmed.douik@caltech.edu, $^\ddagger$samehsorour@uidaho.edu}
\vspace{-.8cm} }

  \maketitle

\global\csname @topnum\endcsname 0
\global\csname @botnum\endcsname 0

\begin{abstract}
This paper is interested in maximizing the total throughput of cloud radio access networks (CRANs) in which multiple radio remote heads (RRHs) are connected to a central computing unit known as the cloud. The transmit frame of each RRH consists of multiple radio resources blocks (RRBs), and the cloud is responsible for synchronizing these RRBS and scheduling them to users. Unlike previous works that consider allocating each RRB to only a single user at each time instance, this paper proposes to mix the flows of multiple users in each RRB using instantly decodable network coding (IDNC). The proposed scheme is thus designed to jointly schedule the users to different RRBs, choose the encoded file sent in each of them, and the rate at which each of them is transmitted. Hence, the paper maximizes the throughput which is defined as the number of correctly received bits. To jointly fulfill this objective, we design a graph in which each vertex represents a possible user-RRB association, encoded file, and transmission rate. By appropriately choosing the weights of vertices, the scheduling problem is shown to be equivalent to a maximum weight clique problem over the newly introduced graph. Simulation results illustrate the significant gains of the proposed scheme compared to classical coding and uncoded solutions.
\end{abstract}

\IEEEpeerreviewmaketitle

\section{Introduction}

The continuously increasing demand for high-speed data transfer over the air imposes severe burdens on current wireless networks, and thus calls for the design of extremely efficient transmission solutions. Moreover, the scarcity of the radio resources raises extra challenges on the next generation wireless networks to meet the expected quality of service requirements by end-users. To address these issues, a revolution in the network's architecture design is required. The move towards dense cellular architectures solved a component of the problem but raised more problems regarding interference management. The solution was recently proposed in \cite{Andrews_2014} which gave birth to cloud radio access networks (CRANs) \cite{Park_13,Cai_2014} in which multiple remote radio heads (RRHs) are coordinated in a centralized fashion by a computing unit known as the cloud. 

With this new architecture, there is a need to design scheduling schemes of RRHs and their associations to users in order to fully utilize its benefits. Without cloud coordination, RRH scheduling in heterogeneous networks was performed using a preassigned association of mobile users and RRHs, e.g., proportional fair scheduling \cite{4,5}. Recent works on CRANs, e.g., \cite{park2013joint,shi2014group}, suggested scheduling users to RRHs in a coordinated fashion at the cloud so as to maximize the network total ergodic capacity. \ignore{The studies are extended to include the joint optimization of the scheduling with the beamforming vectors and the power level of each radio resource block (RRB).} These works, however, view the network solely from the physical layer, e.g., \cite{15}. Therefore, each radio resource block (RRB) serves only one user in each transmission instant. Clearly, this does not take into consideration upper layer facts. For instance, it has been recently found that users tend to have a common interest in downloading popular files (especially videos) within a small interval of time, thus creating a pool of side information in the network. This paper proposes to incorporate network coding (NC) in the scheduling decisions of RRHs and RRBs in order to utilize this side information in enhancing the system throughput when other users request to download these same files.

NC was introduced \cite{850663} as a new paradigm that performs flow mixing  (i.e., coding) at intermediate nodes in the network. It has demonstrated great benefits to improve the performance for numerous network metrics such as throughput improvement and delay minimization \cite{6512065}. A particularly interesting sub-class of NC \ignore{that is well-designed for time-critical applications, e.g., sensors instructions and emergency messages to crew members,} is the instantly decodable network coding (IDNC). Indeed, thanks to its instant decodability properties and straightforward operations to encode/decode files, IDNC was the subject of numerous studies, e.g., \cite{M.}-\cite{598452}. IDNC uses XOR-based operations for encoding at the transmitters and decoding at the receivers. These simple operations are well-adapted to small and battery-powered devices.

Different studies on IDNC revealed various code construction schemes with excellent potential in minimizing various system parameters for different applications and network settings. For example, while the authors in \cite{M.} suggest minimizing the total transmission time, i.e., the completion time, reference \cite{S.} optimizes the decoding delay. Similarly, the authors in \cite{A.} introduce a delay-based framework to reduce the completion time. Recently, IDNC was employed in a heterogeneous network setting to minimize the completion time of the users by jointly selecting the message combinations and transmission rates of each RRH \cite{598452}. This paper aims to extend the study to the CRAN setting. By exploiting the side information of each user, obtained from previous file downloads, coded files can be multicast to several users to increase the total received throughput in each transmission. This work investigates the use of IDNC in CRANs in order to maximize such received throughput.

Given the above facts, the overall received throughput maximization problem of interest in this paper must consider the joint scheduling of users to RRBs in RRHs, choice of the encoded file sent in each of them, and the rate at which each of them is transmitted. For practical CRANs design, e.g., limited capacity backhaul links, the above joint problem should be solved under the constraint that each user can be scheduled to a single RRH, but possibly to multiple RRBs within it. Furthermore, each RRH should transmit at a rate less than the ergodic capacity of all the scheduled users for each RRB. Consequently, increasing the number of multiplexed users in each RRB can easily decrease the transmission rate, and thus the overall system throughput.

The main contribution of this paper is to solve the above received throughput maximization problem. To that end, the paper introduces a graph, called herein the CRAN-IDNC graph, in which vertices represent a 5-tuple combination of RRH, RRB, user, file, and transmission rate. A one-to-one correspondence between the cliques of the graph and the solutions of the scheduling problem is established. By carefully designing the weights of the vertices, the problem of interest is shown to be equivalent to a maximum weight clique problem over the CRAN-IDNC graph. \ignore{Simulation results show the performance of the proposed scheme and reveal that it largely outperforms coding-free solutions.}

The rest of this paper is organized as follows: Section \ref{SMMM} introduces the system model. The network coding model and the problem formulation are illustrated in section \ref{C-RAN}. Section \ref{PS} designs the graph and presents the proposed solution. Finally, before concluding in Section \ref{CC}, Section \ref{SR} plots and discusses the simulation results.

\section{System Model and Parameters} \label{SMMM}

\subsection{Network Model}
\begin{figure}[t]
\centering
\includegraphics[width=0.8\linewidth]{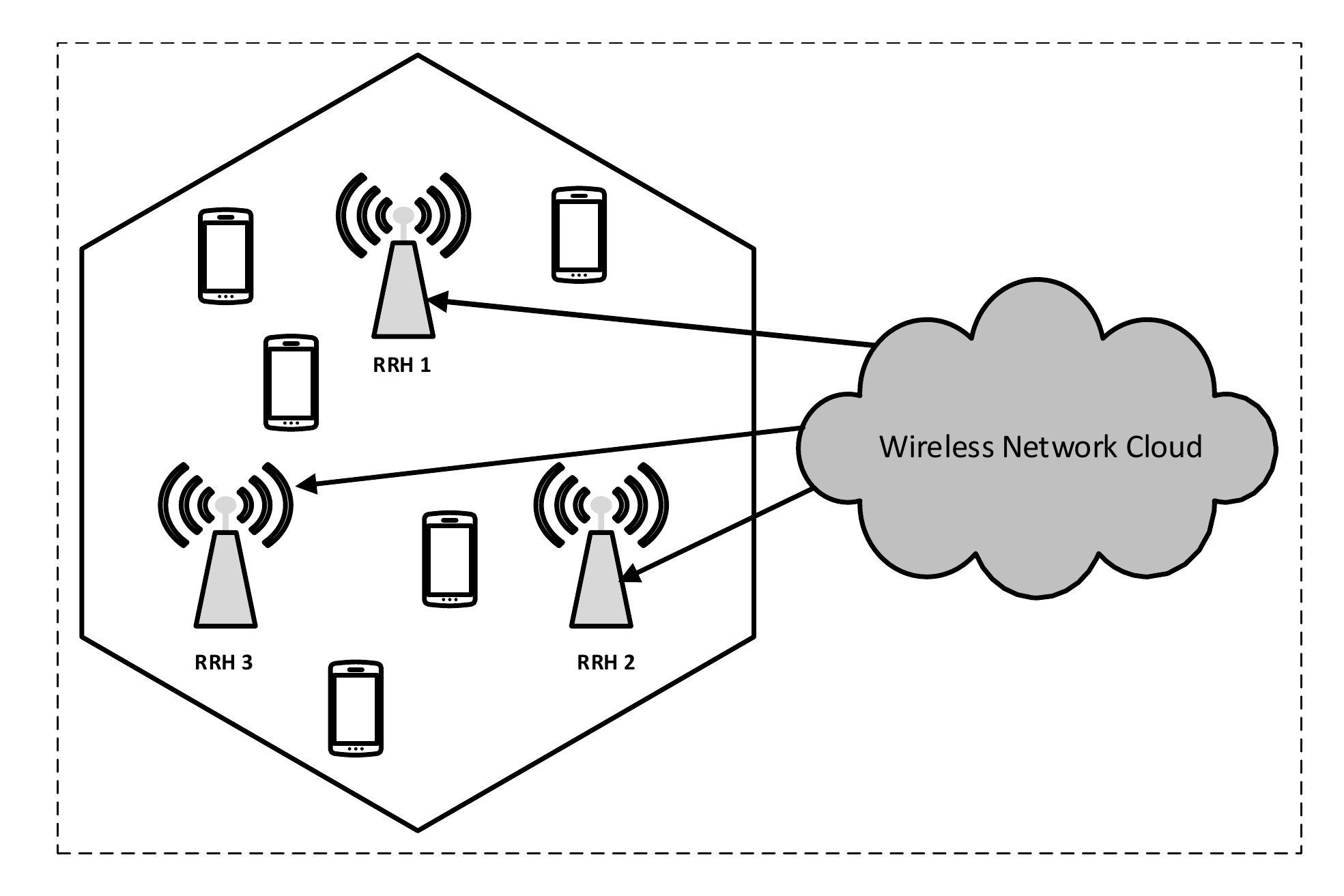}
\caption{A CRAN composed of $5$ users, and $3$ RRHs.}
\label{fig1}
\end{figure}
This paper considers the use of IDNC in the downlink of a CRAN as the one shown in Figure \ref{fig1}. The network consists of a set $\mathcal{B}$ of $B$ remote radio heads (RRHs), which are distributed in different geographic locations within the cell. These $B$ RRHs are connected to a central computing unit, i.e., the cloud controller, which is responsible for the dynamic scheduling, control decisions, and synchronization of all the RRHs' transmit frames. For instance, in Fig. \ref{fig1}, RRHs $1$, $2$ and $3$ cooperate to serve all mobile users in their joint coverage range. Each RRH's transmit frame consists of $Z$ orthogonal time/frequency RRBs as shown in figure \ref{fig2}.

Let $\mathcal{Z}$ be the set of RRBs in the frame of each RRH. The transmit power $P_{bz}$ of the $z$-th RRB in the $b$-th RRH can be typically different from those of other RRBs, and is maintained at a fixed value within the RRB's duration. The RRB power levels are typically updated in an outer loop, but this falls outside the scope of this paper and is left of future investigations. The total number of available RRBs is the total number of RRH times the number of RRBs in each RRH, that is $Z_{tot}=BZ$.

The network serves a set $\mathcal{U}$ of $U$ mobile users possessing side information. The cloud thus exploits this side information to allocate users, which can be simultaneously satisfied by one coded transmission to the same RRBs of one RRH (See Section III.A for more details). To reduce handover complexities, each user can be assigned to at most a single RRH, but possible to many RRBs within it. The use of coding in the scheduling decisions allows the mixing of multiple users' requests in one RRB, provided that their requests can be all served by one coded file, and that file is transmitted at a suitable rate, given the RRB's transmit power, for all of them to receive it correctly.
\begin{figure}[t]
\centering
\includegraphics[width=0.6\linewidth]{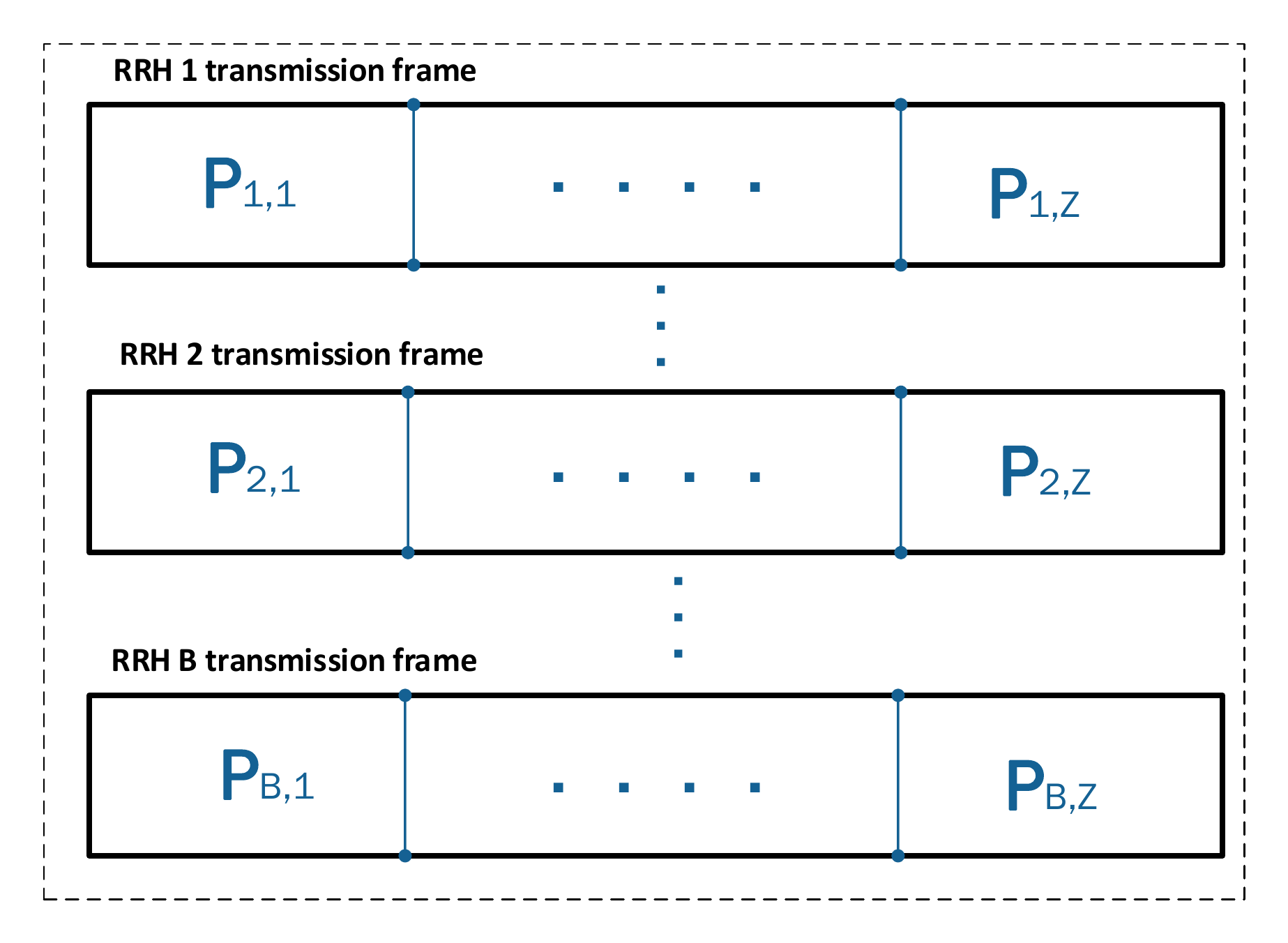}
\caption{Frame structure of the $B$ RRHs each containing $Z$ RRBs.}
\label{fig2}
\end{figure}

\subsection{Physical Layer Model}

Let $h_{bz}^{u}$ be the complex channel gain from the $z$-th RRB in the $b$-th RRH to the $u$-th user. The channel gain is assumed to remain constant during the transmission time of a single file. This paper considers a general model for the channels without restrictions on their distribution. However, it uses the standard assumption that these values are perfectly estimated and available at the cloud. The ergodic capacity of the $u$-th user assigned to the $z$-th RRB in the $b$-th RRH can be expressed as:
\begin{equation}
R^{u}_{bz}(t)= \log_{2}(1+\text{SINR}^{u}_{bz}),
\end{equation}
wherein SINR$^{u}_{bz}$ is the corresponding signal-to-interference plus noise-ratio experienced by the $u$-th user when it is associated with RRB $z$ of the $b$-th RRH. This SINR can be expressed by the following formula:
\begin{equation}
\text{SINR}^{u}_{bz} = \cfrac{P_{bz} |h^{u}_{bz}|^{2}}{\sigma^2 +\sum\limits_{b'\neq b}P_{b'z}| h^{u}_{b'z}|^{2}},
\end{equation}
where $\sigma^{2}$ is the complex Gaussian noise variance. Note that this paper assumes orthogonal RRBs within the same RRH. As a result, interference at the $z$-th block is only seen from RRBs that have the same index $z$ across the other RRHs. The reception of an encoded file sent in the $z$-th RRB of the $b$-th RRH is successful at the $u$-th mobile user if the transmission rate $R_{bz}$ is less than or equal the user's capacity, i.e., $R_{bz}\leqslant R_{bz}^u$. In other words, the $z$-th RRB of $b$-th RRH can transmit at a rate at most equal to the minimum ergodic capacity of its assigned users. It is important to note that this paper assumes perfect channel estimation scenario so that the $t$-th file transmission by $z$-th RRB of $b$-th RRH to its assigned users is received successfully without any erasures.

The set of achievable capacities of all users in all RRBs across all RRHs can be represented by the set:
\begin{align}
\mathcal{R} = \bigotimes_{(b,z,u) \in \ \mathcal{\mathcal{B}} \times \mathcal{Z} \times \mathcal{\mathcal{U}}} R^u_{bz}, \label{eqw3}
\end{align} 
wherein the symbol $\bigotimes$ represents the product of the set of the achievable capacities. 

\section{Network Coding and Problem Formulation} \label{C-RAN}

This section first introduces the side information model and IDNC. Afterward, the overall throughput maximization problem in CRANs is formulated.

\subsection{Instantly Decodable Network Coding}

The paper assumes that the users are interested in receiving one or more files out of a set of $\mathcal{{F}}$ of $F$ files, which are deemed popular due to their previous multiple downloads by different subsets of the users. All files in $\mathcal{{F}}$ are assumed to have the same size of $N$ bits, and thus, an XOR encoding of any number of files (that we will refer to as an encoded file) is also $N$ bits.\ignore{ Due to the dynamic capacities of users and the erasure nature of the channel links, the successive transmissions create an asymmetric side information at each user. Indeed, at each transmission slot, the set $\mathcal{{F}}$ can be decomposed for $u$-th user into the following sets:} The different prior downloads of the users from $\mathcal{F}$ creates an asymmetric side information in the network. Indeed, in each scheduling epoch, the files of $\mathcal{F}$ can be classified for each user $u$ as follows:
\begin{itemize}
\item The \textit{Has} set $\mathcal{H}_{u}$ containing the files previously downloaded by the $u$-th user.
\item The \textit{Wants} set $\mathcal {W}_u\subseteq\mathcal{F}\backslash \mathcal {H}_u$ containing the file(s) requested by the $u$-th user in the current scheduling frame.
\end{itemize}
It is typical in this scenario that the cloud keeps a log of all downloaded files by the users, and thus the side information is perfectly available to the cloud\ignore{ which collects feedback and acknowledgments through the different RRHs}. This side information can be thus exploited to transmit encoded files instead of sending simple uncoded files. Consequently, the cloud controller is thus required to jointly perform the users-RRH/RRB scheduling, encoded file selection, and rate selection for each RRB, in order to maximize the number of received decodable bits, i.e., the throughput, in each scheduling frame.

\begin{figure}[t]
\centering
\includegraphics[width=0.6\linewidth]{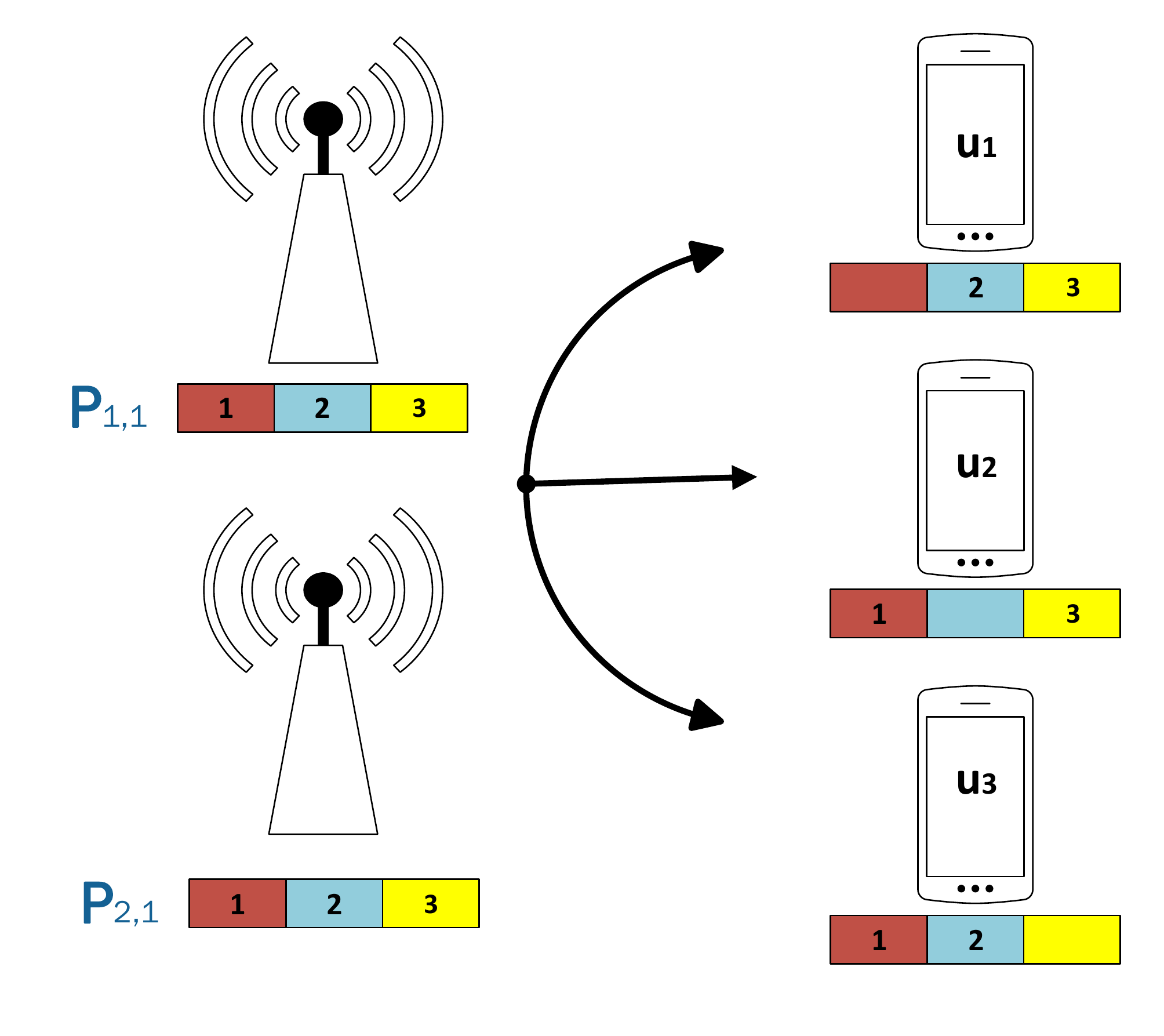}
\caption{A CRAN composed of $3$ users, $3$ files,  $2$ remote radio
heads, and $1$ RRB in each RRH's transmit frame.}
\label{fig3}
\end{figure}

IDNC allows the cloud to generate XOR-encoded files using the source files in $\mathcal{F}$. Therefore, users that receive a combination that contains only \emph{one} file from its Wants sets are able to decode the combination and retrieve their wanted file. A combination that can be used to extract a new wanted file by any user is said to be instantly decodable for that user.

Consider Figure \ref{fig3} which illustrates an example of a CRAN system composed of $3$ users, $3$ files, $2$ RRHs, and $1$ RRB per RRH frame. Assuming that the achievable capacities of all users to the RRB is $1$ bit/second. The file combination $1\oplus2$ is instantly decodable for users $1$ and $2$ but not for $3$. The optimal achievable overall throughput in this scenario is $3$ bits/s. Indeed, the first RRH targets users $1$ and $2$ with $1\oplus2$ and the second serves user $3$ with $3$. Such approach improves upon the achievable $2$ bits/s throughput without coding.

\subsection{Problem Formulation}

This paper considers maximizing the throughput in the aforementioned CRAN setting by assigning users to the RRBs in the RRHs under the following network connectivity constraints (CC):
\begin{itemize}
\item \textbf{CC}: Each mobile user can connect to at most a single RRH, but possibly to many RRBs in that RRH.
\end{itemize}

Let $\kappa_{bz}$ be the transmitted (encoded) file in the $z$-th RRB of the $b$-th RRH, i.e, $\kappa_{bz}$ is an element of the power set $\mathcal{P}(\mathcal{F})$. Let $X_{bz}^u$ be a binary variable that is set to $1$ if user $u$ is assigned to the $z$-th RRB of the $b$-th RRH, and zero otherwise. Let $Y_{b}^u$ be a binary variable that is set to $1$ if user $u$ is assigned to the $b$-th RRH, and zero otherwise. The overall throughput maximization problem can be formulated as follows:
\begin{subequations}
\label{Cache_Scheduling_Problem}
\begin{align}
&\max\sum _{b\in \mathcal{\mathcal{B}}} \sum _{z \in \mathcal{\mathcal{Z}}} \sum _{u\in\mathcal{\tau}_{bz}(\kappa_{bz})}X_{bz}^u R_{bz} \label{eq4a} \\
 &{\rm s.t.\ } Y_{b}^u=\mathcal \min\left( \sum _{z}X_{bz}^u,1 \right),~\forall~ (b,u)\in \mathcal{B} \times \mathcal{U}   , \label{eq4b} \\
& \sum _{b}Y_{b}^u\leqslant1,~\forall~ u\in \mathcal{U}, \label{eq4c} \\
& \tau_{bz}(\kappa_{bz}) = \left\{u \in \mathcal{U}\ \bigg|\kappa_{bz} \cap \mathcal{W}_u = 1  \mbox{ \& } R_{bz}\leqslant R^u_{bz}\right\}, \label{eq4d}\\
&  X_{bz}^u , Y_{b}^u\in \{0,1\},{\kappa}_{bz}\in \mathcal{P}(\mathcal{F}), (u,b,z)\in \mathcal{U}\times \mathcal{B}\times\mathcal{Z},\label{eq4e}
\end{align}
\end{subequations}
The optimization is carried over the variables $X_{bz}^u$, $Y_{b}^u$, $\kappa_{bz}$, and $R_{bz}$. The variables $Y_{b}^u$ and $X_{bz}^u$ are discrete optimization parameters that represent the user-RRH and user-RRB associations, respectively. On the other hand, the variables $\kappa_{bz}$ and $R_{bz}$ account for the file combination and the transmission rate for the $z$-th RRB of the $b$-th RRH, respectively. Constraints \eqref{eq4b} and \eqref{eq4c} translate the system condition \textbf{CC}, i.e., each user must connect to at most one RRH. The variable $\tau_{bz}(\kappa_{bz})$ denotes the targeted set of users benefiting from the encoded file transmitted in the $z$-th RRB of the $b$-th RRH. Consequently, Constraint \eqref{eq4d} ensures that all users belonging to these targeted sets $\tau_{bz}(\kappa_{bz})$ $\forall~b\in\mathcal{B}$ and $z\in\mathcal{Z}$ must receive an instantly decodable transmission. In other words, the targeted users can retrieve and decode a new file from the combination $\kappa_{bz}$ when transmitted at the rate $R_{bz}$. 

Being a mixed discrete-continuous optimization problem, the throughput maximization problem expressed in (\ref{eq4a}) may require an extensive search over all possible user-RRH/RRB associations, file combinations, and determining the achievable capacities for every possible association in each RRB across all RRHs. The rest of this paper proposes a more efficient algorithm to solve the optimization problem in (\ref{eq4a}).

\section{Proposed Throughput Maximization Solution}\label{PS}

In order to efficiently solve the mixed integer-continuous problem, this paper uses graph theory techniques to map the feasible points to cliques in a newly introduced CRAN-IDNC graph. The graph is formed by several clusters (called herein RRB-IDNC subgraphs) each corresponding to one RRB in the network. Furthermore, by carefully designing the weights of the vertices, the joint optimization problem is shown to be equivalent to finding the maximum weight clique over the CRAN-IDNC graph. Finally, the optimal and low-complexity graph-theoretic algorithms available in the literature, e.g., \cite{15522856,2155446,6848102,5341909}, can be exploited to find the solution.

In order to construct the CRAN-IDNC graph, this section first creates its building block known as the RRB-IDNC subgraph. Such subgraph is generated for each RRB $z$ in each RRH $b$ with a total of $Z_{tot}=B \times Z$ local subgraphs to be created. The CRAN-IDNC graph is constructed by taking the union of all these subgraphs and appropriately determining the edges between their vertices.

Let $\mathcal{A}$ be the set of all possible associations between RRHs, RRBs, users, files, and achievable capacities, i.e., $\mathcal{A}=\mathcal{B}\times\mathcal{Z}\times \mathcal{U}\times \mathcal{F}\times\mathcal{R}$. Let $\varphi_b$, $\varphi_z$, $\varphi_u$, $\varphi_f$, and ,$\varphi_r$ be mappings from the set $\mathcal{A}$ to the set of RRHs $\mathcal{B}$, the set of RRBs $\mathcal{Z}$, the set of users $\mathcal{U}$, the set of files $\mathcal{F}$, and the set of achievable capacities $\mathcal{R}$. In other words, given the element $y=(b,z,u,f,R) \in \mathcal{A}$, the mappings are defined as $\varphi_b(y)=b$, $\varphi_z(y)=z$, $\varphi_u(y)=u$, $\varphi_f(y)=f $ and $\varphi_r(y)=R$, respectively. Finally, let $\mathcal{P}(\mathcal{A})$ represent all possible schedules of associations between RRHs, RRBs, users, files, and the achievable capacities regardless of the feasibility of each, i.e., whether each schedule satisfy \textbf{CC} or not. Similarly, let $\mathcal{A}_{bz}$ represent the association relative to the $z$-th RRB in the $b$-th RRH, i.e., $y \in \mathcal{A}_{bz} \Rightarrow \varphi_b(y)=b,$ and $\varphi_z(y)=z$.
\begin{table}[!t]
\renewcommand{\arraystretch}{1}
\caption{Variables and Parameters of the System}
\label{table_example1}
\centering
\begin{tabular}{|l|l|}
\hline
$\mathcal{U}$ & Set of $U$ mobile users\\
\hline
$\mathcal{F}$ & Set of $F$ files \\
\hline
$\mathcal{B}$ & Set of $B$ Radio Remote Heads (RRH) \\
\hline
$\mathcal{Z}$ & Set of $Z$ Radio Resource Blocks (RRBs)\\
\hline
$\mathcal{R}$ & Set of all achievable capacities \\
\hline
$\mathcal{W}_u$ & Set of wanted files by user $u$ \\
\hline
$R_{bz}$ & Transmission rate of $z$ RRB in $b$ RRH \\
\hline
$\kappa_{bz}$ & The encoded file of the $z$-th RRB \\ & in the $b$-th RRH \\
\hline
$\tau_{bz}(\kappa_{bz})$ & Set of targeted users by  $\kappa_{bz}$  \\
\hline
$\mathcal{A}$ & Set of all possible associations \\
\hline
$\mathcal{A}_{bz}$ & The association of $z$ RRB in $b$ RRH \\
\hline

\end{tabular}
\end {table}

\subsection{Construction of the RRB-IDNC Subgraphs}
Let the RRB-IDNC subgraph of RRB $z$ in RRH $b$ be denoted by $\mathcal{G}^{bz}(\mathcal{V}^{bz},\mathcal{E}^{bz})$ wherein $\mathcal{V}^{bz}$ and $\mathcal{E}^{bz}$ refer to the set of vertices and edges of this subgraph, respectively. This RRB-IDNC subgraph is constructed by generating a vertex $v$ for each possible association $ s \in \mathcal{A}_{bz}$. Two vertices $v \in \mathcal{V}^{bz}$ associated with $s \in \mathcal{A}_{bz}$ and $v^\prime \in\mathcal{V}^{bz}$ associated with $s^\prime \in\mathcal{A}_{bz}$ are adjacent by an edge, if both following local conditions (LC) are true:
\begin{itemize}
\item {\textbf{LC1}:} $\left(\varphi_f(s)\in \mathcal{H}_{\varphi_u(s')}~\mbox{and}~ \varphi_f(s^\prime)\in \mathcal{H}_{\varphi_u(s)}\right)$ or $\left(\varphi_f(s)= \varphi_f(s')\right)$. This condition shows that the file combination is immediately decodable for both users $\varphi_u(s)$ and $\varphi_u(s^\prime)$ as they either have each other's requested file or are requesting the same file, respectively. 
\item {\textbf{LC2}}: $\varphi_r(s)=\varphi_r(s^\prime)$. This condition guarantees that all adjacent vertices in the subgraph have the same transmission rate.
\end{itemize}

\subsection{Construction of the CRAN-IDNC Graph}

As stated earlier, the CRAN-IDNC graph $\mathcal{G}(\mathcal{V},\mathcal{E})$ is constructed by first generating all the $Z_{tot}$ RRB-IDNC subgraphs. The vertex set of the CRAN-IDNC graph is simply the union of the vertices of all the RRB-IDNC subgraphs; i.e., $\mathcal{V} = \bigcup_{b \in \mathcal{B}}\bigcup_{z \in \mathcal{Z}}\mathcal{V}^{bz}$. The edges between vertices within the same RRB-IDNC subgraph are already described in the previous subsection. Two different vertices belongs to two different subgraphs are then set adjacent if their combination results in a feasible schedule, i.e., it satisfies the system constraint \textbf{CC}. To express this fact mathematically, let vertex $v\in \mathcal{G}^{bz}$ be corresponding to the association $s \in \mathcal{A}_{bz}$ and vertex $v^\prime\in \mathcal{G}^{b^\prime z^\prime}$ corresponding to the association $s^\prime \in \mathcal{A}_{b^\prime z^\prime}$. The vertices $v$, $v^\prime$ are adjacent if they satisfy one of the following general conditions: 
\begin{itemize}
\item {\textbf{GC1}:} $\varphi_u(s)=\varphi_u(s')$ and $\varphi_b(s)=\varphi_b(s')$, $\forall~(s,s')\in \mathcal{A}_{bz}\times \mathcal{A}_{b^\prime z^\prime}$. This condition translates the fact that the same user can be served with multiple RRBs within the same RRH. 
\item {\textbf{GC2}:}  $\left(\varphi_b(s) =\ \varphi_b(s^\prime)~\mbox{and}~\varphi_f(s)= \varphi_f(s')\right)$ OR $\left(\varphi_b(s) =\ \varphi_b(s^\prime)~\mbox{and}~\varphi_f(s)\in \mathcal{H}_{\varphi_u(s')}~\mbox{and}~ \varphi_f(s^\prime)\in \right.$ $\left.\mathcal{H}_{\varphi_u(s)}\right)$. This condition guarantees that the encoded combinations of the same users can be served by multiple RRBs within the same RRH.
\item {\textbf{GC3}:} $\varphi_u(s) \neq\ \varphi_u(s^\prime)$ and $\left(\varphi_b(s), \varphi_z(s)\right) \neq\ \left(\varphi_b(s^\prime),\varphi_z(s')\right)$. This condition completes the adjacencies in the graph for any two vertices not opposing the CC constraint for any two different users. 
\end{itemize} 

An example of the CRAN-IDNC graph is shown in Fig. \ref{fig4} for a simple network consisting of $3$ users, $3$ files, $2$ RRHs, and $1$ RRB in each RRH frame, i.e., the system in figure \ref{fig3}. In this example, each vertex is labeled $bzufr$, where $b$, $z$, $u$, $f$, and $r$ represent the indices of RRHs, RRBs, users, files, and achievable capacities respectively. The Dashed and solid lines in Fig. \ref{fig4} represent the edges generated by the aforementioned conditions, and the potential maximal cliques in this example represented in the graph by slid lines are: $\{11111,21221\}$, $\{21332,11111\}$, $\{21221,11332\}$, $\{21111,21221,21331\}$, $\{11111,11221,11331\}$, $\{21112,11222,11332\}$, and $\{11223,21112,21332\}$, achieving a total throughputs of 2, 3, 3, 3, 3, 6, and 7 bits/sec, respectively. Clearly, the last maximal clique should be the one selected as it maximizes the throughput of the entire CRAN for this scheduling frame\ignore{, but serves all three users simultaneously (as opposed to cliques 1, 2, 3, 6, and 7 according to the above order)}.

\begin{figure}[t]
\centering
\includegraphics[width=1\linewidth]{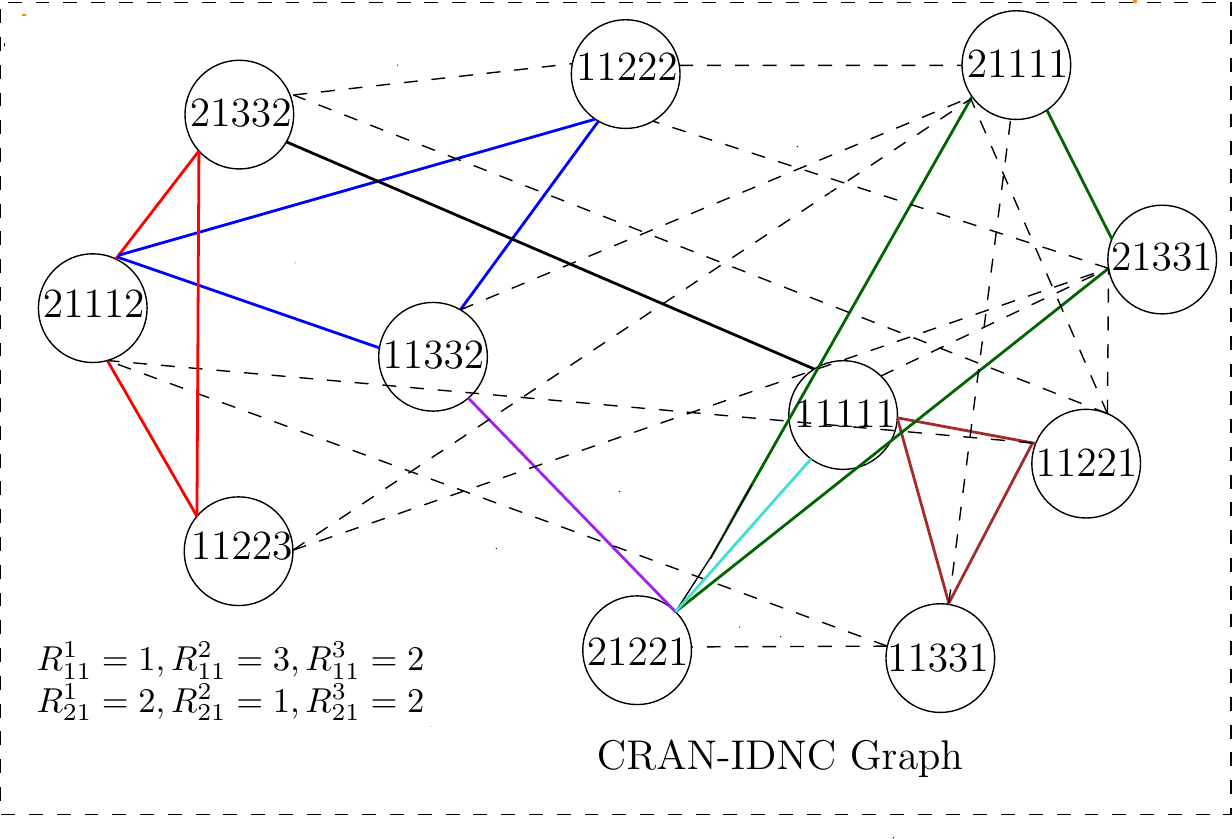}
\caption{The CRAN-IDNC graph of the network presented in Figure \ref{fig3}.}
\label{fig4}
\end{figure}

\subsection{Optimal Assignment Solution}

Given the above construction of the CRAN-IDNC graph $\mathcal{G}(\mathcal{V},\mathcal{E})$, it can be established that any maximal clique in the graph satisfies the following criterion:
\begin{itemize}
\item All users with vertices in the maximal clique can decode a new file from the transmission schedule of all RRBs and RRHs.
\item The transmission rate of each RRB identified by the vertices in a maximal clique is less than or equal to the channel capacities of all the users having vertices with the same RRB in that clique. 
\end{itemize}

The following theorem characterizes the solution to the throughput maximization problem of interest in this paper.
\begin{theorem}
The CRAN throughput maximization problem in (\ref{eq4a}) is equivalent to a maximum weight clique problem over the CRAN-IDNC graph, wherein the weight of a vertex $v \in \mathcal{V}$ corresponding to the association $s=(b,z,u,f,r) \in \mathcal{A}$ is given by:
\begin{align}
w(v) = r.
\end{align}
The set of targeted users and the file combination of the $z$-th RRB in the $b$-th RRH is obtained by combining the vertices of the maximum weight clique corresponding to the local RRB-IDNC graph $\mathcal{G}^{bz}$.
\end{theorem}

\begin{proof}
This theorem is proved by demonstrating the following facts: The first fact establishes a one-to-one mapping between the feasible schedules and the cliques in the CRAN-IDNC graph. Afterwards, the weight of each vertex is set to be the contribution of the corresponding user to the network. Therefore, the maximum weight clique is a feasible solution with the maximum received-throughput. In other words, the maximum weight clique is the solution to (\ref{eq4a}). The complete proof can be found in Appendix \ref{APA}. 
\end{proof}

Recall that a clique in a graph is a set of vertices such that each pair is connected. The maximum weight clique problem is the one of finding the clique with the maximum weight. The maximum weight clique problem is an NP-complete problem, and even its approximation is hard \cite{16}. However, it can be optimally solved with a reduced complexity as compared to the naive search, e.g., the optimal algorithms in \cite{15522856,2155446}. Furthermore, multiple efficient polynomial time heuristics \cite{6848102,5341909} are proposed in the literature. This paper uses the simple quadratic complexity heuristic proposed in \cite{15}.

\section{Simulation Results} \label{SR}

\begin{table}[!t]
\renewcommand{\arraystretch}{1}
\caption{Simulation Parameters}
\label{table_example}
\centering
\begin{tabular}{|c|c|}
\hline
Cellular Layout & Hexagonal \\
\hline
Cell Diameter & 500 meters \\
\hline
Number of Users & Variable \\
\hline
Number of RRBs & Variable \\
\hline
Channel Model & SUI-Terrain type B \\
\hline
Channel Estimation & Perfect \\
\hline
High Power & -42.60 dBm/Hz\\
\hline
Background Noise Power &-168.60 dBm/Hz\\
\hline
Bandwidth &10 MHz\\
\hline
\end{tabular}
\end {table}

\begin{figure}[t]
\centering
\includegraphics[width=0.9\linewidth]{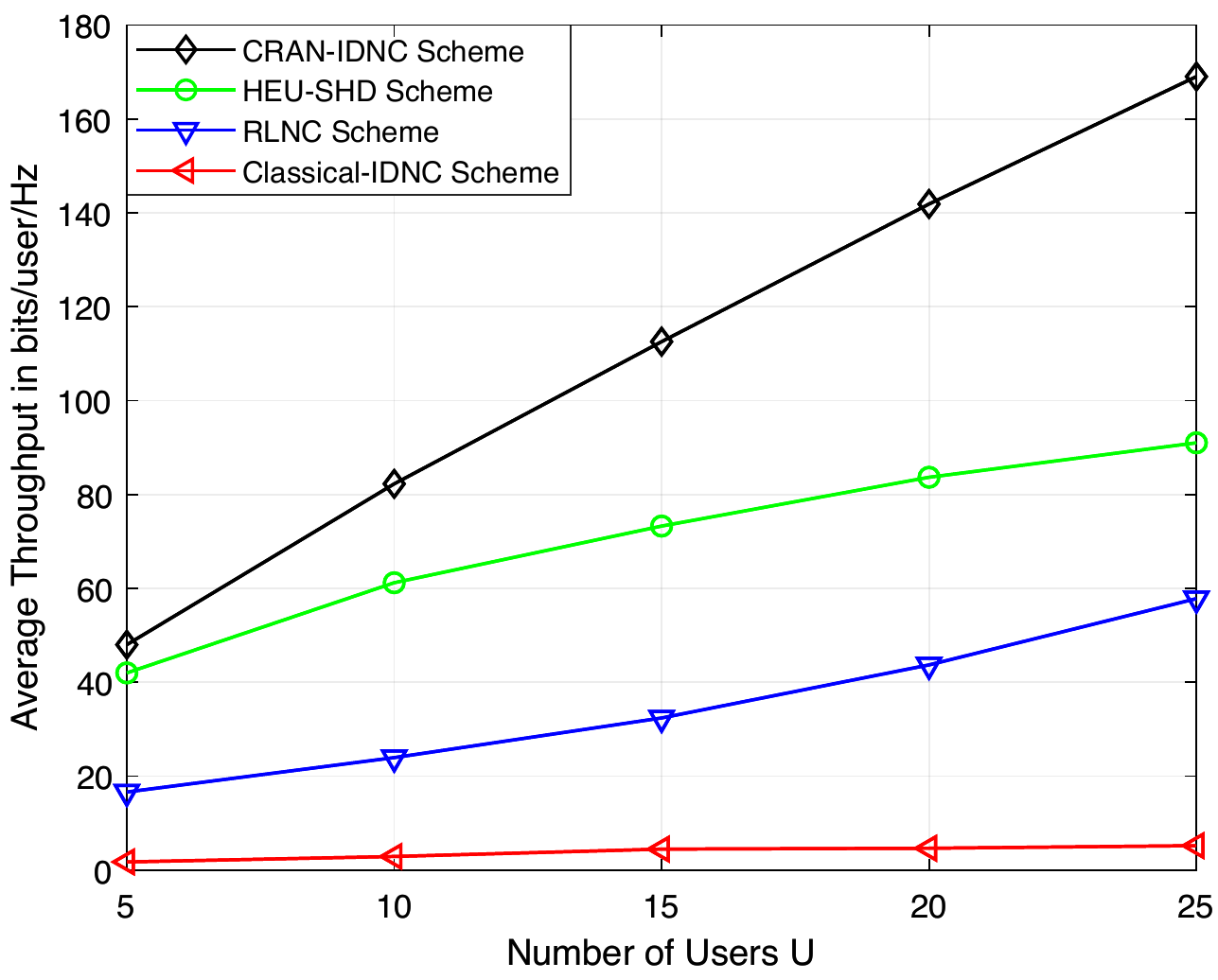}
\caption{Average Throughput in bits/user/Hz. Vs the number of users $U$.}
\label{fig5}
\end{figure}

\begin{figure}[t]
\centering
\includegraphics[width=0.9\linewidth]{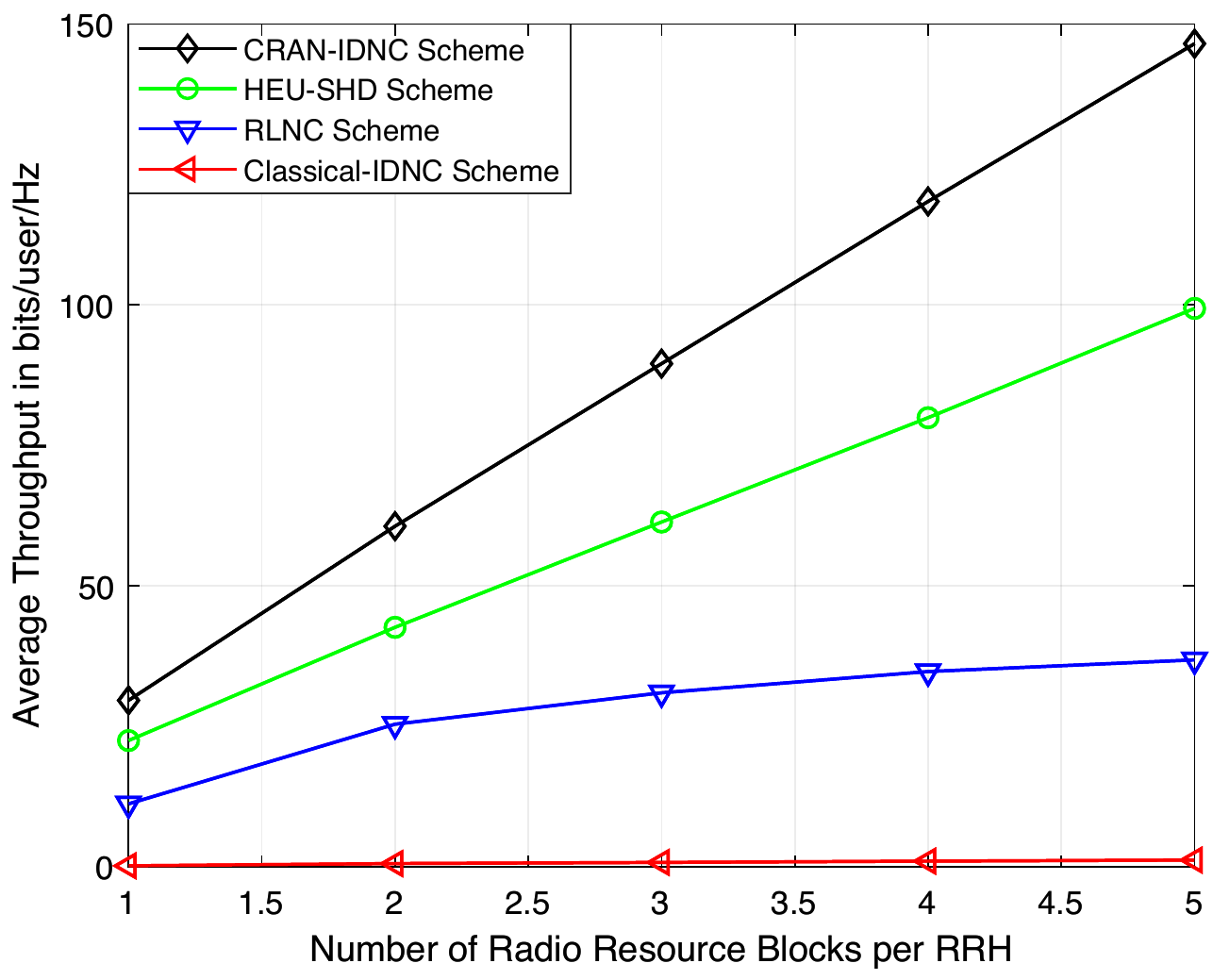}
\caption{Average Throughput in bits/user/Hz. Vs the number of Radio Resource Blocks $Z$.}
\label{fig6}
\end{figure}

This section illustrates the performance of the proposed solution in the CRAN described in Section II. The RRHs are uniformly placed in a hexagonal cell. Users are randomly distributed within the cell whose distance is set to $500$ meters. The number of users, number of RRBs and distribution of the side information varies so as to study multiple scenarios. The total number of remote radio heads is fixed to $3$. Table \ref{table_example} summarized the additional simulation parameters.

The performance of the proposed solution is compared to state-of-the-art coded and uncoded methods.\ignore{ As broadcast and unicast schemes, have been previously shown \cite{598452},\cite{20} to perform poorly against coded schemes, this paper omits them comparison.} In particular, this section adopts the following schemes (in a similar manner as \cite{20},\cite{21}) as a baseline algorithms for comparison:
\begin{itemize}
\item Classical IDNC (rate-unaware scheme): This scheme jointly optimizes the selection of an XOR file combination for each RRB in each RRH without considering the achievable capacities of users. After the file selection process, the CRAN's physical-layer employs the minimum achievable capacity of all users targeted by each RRB as its transmitting rate.
\item RLNC (rate-greedy scheme): In this scheme, each user is associated with a single RRB to which it has the maximum capacity. If more than one user is associated with the same RRB, random linear network coding (RLNC) is employed. The encoding is done irrespectively of the side information. Indeed, as stated earlier, RLNC mixes all files with different random coefficients. The selected transmission rate in each RRB is the capacity of users having the minimum achievable capacity in that RRB.
\item HEU SHD scheme (uncoded scheme): In this scheme, only one user is served in each RRB, each user can be assigned to more than one RRBs from the same RRH. This scheme is proposed in \cite{15}.\ignore{ so as to maximize the transmit sum rate of the CRAN.} 
\end{itemize}

Figure \ref{fig5} depicts the average throughput in bits/user/Hz achieved by our proposed and the aforementioned schemes for different numbers of users $U$, given $4$ RRBs per RRH, a file size $N=1$ Mb. From the figure, we note that our proposed CRAN-IDNC scheme greatly outperforms all other schemes. In particular, HEU-SHD only focuses on the high achievable rate at the expense of transmitting at most one message to a single user from each RRB in all RRHs, i.e., a maximum number of targeted users is $Z_{tot}$. On the other hand, RLNC scheme serves a large number of users in each transmission but sacrifices the rate optimality. One can also notice that the gap between our proposed scheme and the HEU-SHD scheme increases as the number of users increases. This can be explained by the fact that the maximum number of targeted users in HEU-SHD is $Z_{tot}$, so increasing number of users beyond the $Z_{tot}$ would not largely improve the performance of the scheme. The proposed scheme, however, benefits from increasing number of users by mixing the flows of more and more users to the same RRB given their instant decodability of their required files due to their side information. 

Figure \ref{fig6} shows the average throughput in bits/user/Hz achieved by different schemes for different numbers of RRBs $Z$, for $U$ = $15$ users and a file size $N=1$ Mb. Again, the figure shows that our proposed CRAN-IDNC scheme outperforms all other schemes. The gap in performance increases as the number of RRBs per frame grows. It can also be easily seen from the figure that the performances of both the proposed and the HEU-SHD schemes increase linearly with the increasing number of RRBs and a fixed number of users. In fact, both schemes meet in serving the same user in different RRBs of the same RRH. Therefore, increases the number of RRBs increases the total received throughput.
\begin{figure}[t]
\centering
\includegraphics[width=0.9\linewidth]{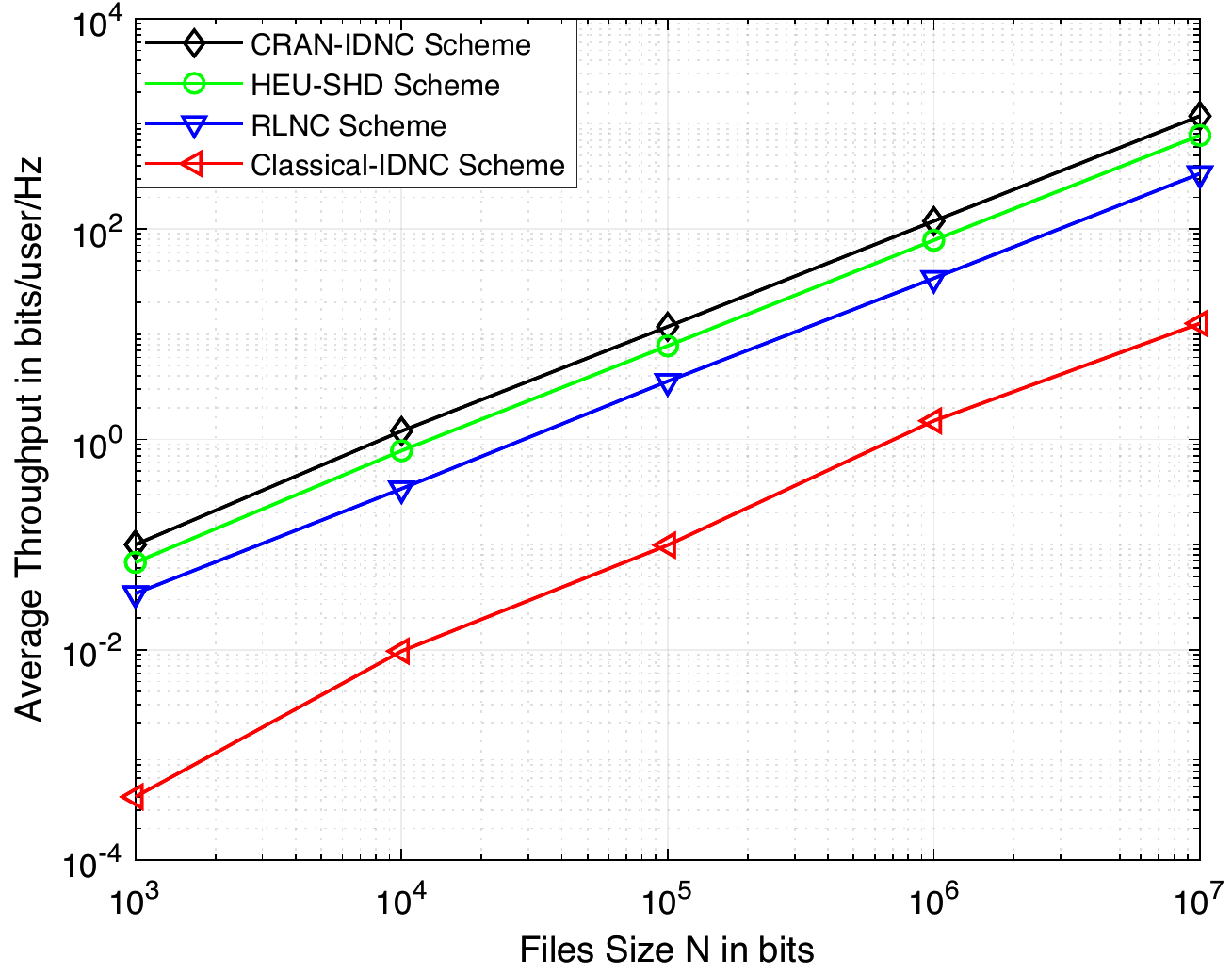}
\caption{Average Throughput in bits/user/Hz. Vs the file's size $N$ in bits.}
\label{fig7}
\end{figure}

Figure \ref{fig7} plots the average throughput as a function of the file size $N$\ignore{ in bits} in a network composed of $15$ users and $4$ RRBs per RRH's transmit frame. As the file's size increases, the performance of all schemes increases. The figure shows that all schemes increase linearly with the size of the file. This can be explained by the fact that, as the size of the file increases, more and more bits are received, thus increasing the average received throughput. Finally, despite its great merits in reducing the completion time of a frame of files in many prior works, Classical IDNC exhibits a very poor performance from a physical layer perspective, thus voiding its upper layer gains. This clearly shows the importance of rate-awareness in IDNC code design to achieve gains\ignore{ in performance} on both the upper and physical layers, thus leading to a real reduction in file delivery times.

\section{Conclusion}\label{CC}
This paper addresses the throughput maximization problem in cloud-enabled networks. Unlike previous studies that only considered the system from a physical layer point of view, we proposed to use the information available in the network to combine files using instantly decodable network coding. Therefore, the throughput maximization problem becomes same as the problem of assigning users to the available resource blocks and choosing the file combination and the transmission rate of each under the constraint that a user can connect to at most a single remote radio head. We have used a graph theoretical approach to solve the problem by introducing the CRAN-IDNC graph formed by several RRB-IDNC graphs. By establishing a correspondence between the feasible solution to the problem and the cliques in the graph, the problem is shown to be equivalent to a maximum-weight clique which can be efficiently solved using state of the art methods. Simulation results show the performance of the proposed scheme and reveal that it largely outperforms coding free solutions.

\appendices
\section{Proof of theorem 1} \label{APA}
\numberwithin{equation}{section}
\setcounter{equation}{0}
This theorem is proofed by showing a one-to-one mapping
between $\mathcal{C}$ (i.e., the set of  cliques) in the
CRAN-IDNC graph and the $\mathbf{S}\in\mathcal{P}(\mathcal{A})$, i.e., feasible schedule that satisfy \textbf{CC}. Afterward, the weight of the clique is shown to be equivalent to the optimization problem (\ref{eq4a}).
Therefore,  we first show that, for each maximal clique in CRAN DNC graph, there is a unique feasible schedule. Then, we show the converse.
The authors in \cite{20} show that there exists a one-to-one
mapping between the set of feasible transmissions and the
set of maximal cliques in the RRB-IDNC graph, i.e., satisfy \textbf{LC1} and \textbf{LC2}.
Hence, to extend the results of \cite{20} to the CRAN- IDNC
graph, we only need to show the following: \begin{itemize}
\item 
The feasible transmissions from different RRHs are adjacent, i.e., the
constraint \textbf{GC1}.
\item
The same user can be targeted from different RRBs within the same RRH frame, i.e., the
constraint \textbf{GC2}.
\item
Two different users  with two different RRBs should be connected, i.e., the
constraint \textbf{GC3}.

\end{itemize}
From the general condition \textbf{GC1} of the CRAN-IDNC graph, the same user cannot be targeted by distinct RRHs. Thus, the vertices representing RRH $b$ in the RRB-IDNC graphs are connected to the vertices in the RRB-IDNC graphs of RRH $b'$, as long as the targeted users are different. From the general condition \textbf{GC2} of the CRAN-IDNC graph, the same user can connect to multiple RRBs within the same RRH transmit frame.  Moreover, from the general condition \textbf{GC3} of the C-RAN graph, two different users with two different RRBs should be connected. That is, the vertices in the $z$-th RRB-IDNC are connected to the $z'$-th RRB-IDNC as long as users are not the same. Therefore, each feasible association of users, RRHs, RRBs, file combinations, and their transmission
rates is represented by a clique $C\in\mathcal{C}$ within the CRAN-IDNC graph. 

On the other hand, it can readily be seen that each clique represents a feasible schedule $\mathbf{S}$ as it does not violate the
conditions \textbf{LC1, LC2,}  \textbf{GC1}, \textbf{GC2} and \textbf{GC3}. Indeed, for the maximum clique $C\in\mathcal{C}$, the transmission of the combination $\kappa_{bz} = \bigoplus_{s\in C} \varphi_f(s)$ in radio resource block $z$ in remote radio head $b$ at rate $r$ is instantly decodable for all users $\tau_{bz}(\kappa_{bz}) = \bigcup_{s\in C} \varphi_u(s)$. 

To conclude the proof,  the  solution of the optimization problem (\ref{eq4a}) is the maximum-weight clique, where the weight of each vertex can be given  by:
\begin{align}
\begin{split} 
w(v)=r. 
\end {split}
\end{align}

\end{document}